\documentclass{article}
\usepackage{ijcai19}
\PassOptionsToPackage{hyphens}{url}
\usepackage[hidelinks]{hyperref}
\usepackage{balance}

\newcount\Comments  
\newcommand{\kibitz}[2]{\ifnum\Comments=1{\color{#1}{#2}}\fi}

\usepackage{times}
\usepackage{soul}
\usepackage{url}

\usepackage[utf8]{inputenc}
\usepackage[small]{caption}
\usepackage{graphicx}
\usepackage{amsmath}
\usepackage{amssymb}
\usepackage{amsthm}
\usepackage{booktabs}
\usepackage{algorithm}
\usepackage{algorithmic}
\usepackage{graphicx}
\usepackage{subcaption}
\usepackage{mwe}
\usepackage{tikz}
\usetikzlibrary{arrows,automata}
\usetikzlibrary{positioning,arrows.meta,,bending}

\usepackage{thmtools, thm-restate}

\usepackage{dsfont}
\usepackage{color, xcolor}
\usepackage{enumerate}

\usepackage{flushend}

\theoremstyle{plain}
\newtheorem{theorem}{Theorem}

\theoremstyle{definition}
\newtheorem{definition}{Definition}

\newcommand{\location}{L} 
\newcommand{\citet}[1]{\citeauthor{#1}~[\citeyear{#1}]}

\newcommand{\dtype}{\tau}       
\newcommand{\supt}{^{(\dtype)}} 

\newcommand{\dutil}{\pi}

\newcommand{\ra}{\rightarrow}
\newcommand{\im}{\Rightarrow}

\newcommand{\txtst}{~\mathrm{s.t.}~}

\newcommand{\one}[1]{\mathds{1} \{ #1\}}

\newcommand{\obj}{\mathrm{OBJ}}

\newcommand{\E}[1]{\mathbb{E} \left[ #1 \right] }

\title{Ridesharing with  Driver Location Preferences\thanks{The full version of this paper is available at \href{https://128.84.21.199/abs/1905.13191}{arXiv:1905.13191.}
}
}

\author{
Duncan Rheingans-Yoo$^1$%
\and
Scott Duke Kominers$^2$\and
Hongyao Ma$^{3}$\And
David C. Parkes$^3$
\affiliations
$^1$Harvard College\\
$^2$Harvard Business School and Department of Economics\\
$^3$Harvard School of Engineering and Applied Sciences
\emails
\{d\_rheingansyoo@college, kominers@fas, hma@seas, parkes@eecs\}.harvard.edu
}

\begin{document}

\maketitle

\begin{abstract}
We study revenue-optimal pricing and driver compensation in ridesharing platforms when drivers have heterogeneous preferences over locations. If a platform ignores drivers' location preferences, it may make inefficient trip dispatches; moreover, drivers may strategize so as to route towards their preferred locations. In a model with stationary and continuous demand and supply, we present a mechanism that incentivizes drivers to both (i) report their location preferences truthfully and (ii) always provide service. In settings with unconstrained driver supply or symmetric demand patterns, our mechanism achieves the full-information, first-best revenue. Under supply constraints and unbalanced demand, we show via simulation that our mechanism improves over existing mechanisms and has performance close to the first-best.
\end{abstract}

\section{Introduction} \label{sec:intro} 

Uber connected its first rider to a driver in the summer of 2009,\footnote{\url{https://www.uber.com/newsroom/history/}, visited 02/25/2019.} and since then, ridesharing platforms have dramatically changed the way people get around in urban areas. Ridesharing platforms allow a wide array of people to become drivers and---in contrast to traditional taxi systems---use dynamic ``surge pricing" at times when demand exceeds supply. Properly designed, dynamic pricing improves system efficiency~\cite{castillo2017surge}, increases driver supply~\cite{chen2015dynamic}, and makes the system reliable for riders~\cite{hall2015effects}.

A growing literature studies how to structure prices for riders and compensation for drivers so as to optimally account for variation in supply and demand \cite{banerjee2015pricing,Bimpikis2016,castillo2017surge,ma2018spatio}.
However, existing models leave aside driver heterogeneity. In practice, some drivers may 
prefer to drive in the city and others in the suburbs, and many prefer to 
end their days close to home. A matching system that treats drivers as homogeneous makes inefficient dispatches, with drivers preferring to fulfill each other's dispatches instead of their own.

The problem goes beyond simple efficiency loss. A core feature of ridesharing platforms is that drivers retain the flexibility to choose when and where to provide service. Every ride the platform proposes needs to be accepted voluntarily, forming an optimal response for the driver~\cite{ma2018spatio}. This incentive alignment simplifies participation for drivers and also makes behavior more predictable. Without accounting for heterogeneity, a platform cannot fully understand a driver's preferences or  achieve full incentive alignment. 

Indeed, platforms have experimented with methods to incorporate driver heterogeneity. As of Summer 2019, Uber allows drivers to  indicate---twice a day---that they would like to take trips in the direction of a particular location.\footnote{\url{https://help.uber.com/partners/article/set-a-driver-destination?nodeId=f3df375b-5bd4-4460-a5e9-afd84ba439b9}, visited 2/25/19.} However, mechanisms that account for driver preferences can also have unintended consequences if not designed properly. By saying ``I want to drive South," a driver biases her dispatches in a way that could in principle promote more profitable  trips.\footnote{\url{https://therideshareguy.com/uber-drops-destination-filters-back-to-2-trips-per-day/}, visited 2/25/19.}

In this paper, we introduce the study of driver location preference in a mechanism design framework. In Section~\ref{sec:preliminaries}, we adapt a model originally conceived by Bimpikis et al.~\shortcite{Bimpikis2016} to an economy where drivers prefer a particular location. In Section~\ref{sec:parm}, we present the \emph{Preference-Attentive Ridesharing Mechanism} (PARM), which elicits driver preferences and sets a revenue-optimal pricing policy. We show that PARM is incentive-compatible, and that it achieves the full-information, first-best revenue when supply is unconstrained or when demand is symmetric. In Section~\ref{sec:simulations}, we study settings with constrained supply and asymmetric demand, using simulations to compare the revenue and welfare performance of PARM to existing ridesharing mechanisms. We show that PARM achieves close to first-best revenue and typically outperforms even the best case for preference-oblivious pricing (where strategic behavior hurts efficiency). Proofs not presented in the text are deferred to Appendix~A of the full paper.

\subsection{Related Work} \label{sec:related_work}

Existing research on pricing and dispatching in ridesharing platforms does not account for   driver heterogeneity.

We build on work of Bimpikis et al.~\shortcite{Bimpikis2016}, who show that under a continuum model, and with stationary demand and unlimited supply, a ridesharing platform's revenue is maximized when the demand pattern across different locations is balanced. They  show via simulation that in comparison to setting a uniform price for all locations, pricing trips differently depending on trip origins improves revenue. 
Relative to the \citet{Bimpikis2016} model, we allow limited driver supply; moreover, each driver in our model has a preferred location. We thus introduce a reporting phase, in which drivers report their preferred locations. We then modify the matching and pricing formulations in order to align incentives.

\citet{ma2018spatio} study the incentive alignment of drivers in the presence of spatial imbalance and temporal variation of supply and demand. \citet{castillo2017surge} show that dynamic pricing mitigates  inefficient ``wild goose chase'' phenomena for platforms that employ myopic dispatching strategies. Modeling a shared vehicle system as a
continuous-time Markov chain, \citet{banerjee2017pricing} establish approximation guarantees for a static, state-independent pricing policy. \citet{ostrovsky2018carpooling} study the economy of self-driving cars, focusing on car-pooling and market equilibrium.
Queuing-theoretic approaches have also been adopted: \citet{banerjee2015pricing} show the robustness of dynamic pricing; \citet{afeche2018ride} study the impact of driver autonomy and platform control; and \citet{Besbes2018Spatial} analyze the relationship between capacity and performance.

There are various empirical studies, analyzing the impact of dynamic pricing~\cite{hall2015effects,chen2015dynamic}, the labor market for 
drivers~\cite{hall2016analysis,hall2017labor}, consumer surplus~\cite{cohen2016using}, the value of flexible work~\cite{chen2017value}, the gender earnings gap~\cite{cook2018gender}, and the commission vs.~medallion lease-based compensation models~\cite{angrist2017uber}.

\section{Model} \label{sec:preliminaries}

We consider a discrete time, infinite horizon model of a ridesharing network with discrete locations, 
$\location = \{1, \dots, n\}$. Following the baseline model of~Bimpikis et al.~\shortcite{Bimpikis2016}, we assume unit distances, i.e., it takes one period of time to travel in between any pair of locations. 
At the beginning of each time period, for each location $i \in \location$, there is a continuous mass $\theta_i \geq 0$ of riders requesting trips from $i$.
The fraction of riders at $i$ with destination $j \in \location$ is given by $\alpha_{ij} \in [0, 1]$ (thus $\sum_{j \in \location} \alpha_{ij}=1$). We assume that the components of rider demand $\theta = \{\theta_i\}_{i \in \location}$ and $\alpha = \{\alpha_{ij}\}_{i,j \in \location}$ are stationary and do not change over time. 
Riders' willingness to pay for trips are i.i.d.~random variables with CDF $F$. Thus, for any $i,j \in \location$, the number of trips demanded from $i$ to $j$ at price $p_{ij}\geq 0$ would be $\theta_i\alpha_{ij}(1-F(p_{ij}))$. (Riders who are unwilling to pay the stated prices for their rides leave the market.)

Each driver has a preferred location $\tau\in L$. Drivers receive $I \geq 0$ additional utility whenever they start a period in their preferred locations (irrespective of whether they have a rider); this preferred location is private information and represents a driver's {\em type}. For each location $\dtype \in \location$, the total mass of available drivers of type $\dtype$ is given by $s\supt \geq 0$.
Drivers have a discount factor of $\delta \in (0, 1)$, and an outside option that delivers  utility $w \geq 0$.\footnote{Throughout the paper, we consider $\delta$ to be very close to $1$--- this is natural, since an annual interest rate of 4\% implies an exponential discount factor of $0.9999992$ over the course of ten minutes.} 
We assume $\sum_{t = 0}^\infty \delta^t I = I/(1-\delta) < w$, meaning that the utility from being in one's favorite location at all times does not outweigh the outside opportunity.

A {\em ridesharing mechanism} elicits drivers' preferred locations, matches drivers and riders to trips, sets riders' trip prices and drivers' compensation, and (potentially) imposes drivers' penalties for strategic behavior. Before the beginning of the first time period, the mechanism elicits the preferred locations from potential drivers.

At the beginning of each period, a driver whose previous trip ended at location $i$ chooses whether or not to provide service at location $i$. If a driver provides service, the mechanism may dispatch that driver to (i) pick up some rider with trip origin $i$,  (ii) relocate to some location, or  (iii) stay in the same location. 
If a rider going from $i$ to $j$ is picked up by some driver, then the rider pays the platform the trip price $p_{ij} \geq 0$. If a driver of reported type $\dtype$ is dispatched from $i$ to $j$, the platform pays them $c_{ij}\supt \geq 0$, regardless of if her dispatch was to pick up a rider or relocate.%
\footnote{It bears mentioning that $c_{ij}\supt$ need not be a fixed proportion of $p_{ij}$. In fact, Bimpikis et al.~\shortcite{Bimpikis2016} find that for certain types of networks, making driver compensation a fixed proportion of trip price drastically reduces platform revenue.}
Drivers who choose not to provide service in a period can relocate to any location $j$ in the network, are not compensated by the mechanism in this period, and may be charged a penalty $P_j$.\footnote{Drivers only choose whether to provide service at a location and cannot decline dispatches based on the trip destination. This is consistent with current ridesharing platforms, which hide trip destinations because of concern that drivers will cherry pick rides.
}
Denote $p \triangleq \{ p_{ij} \}_{i, j \in \location}$, $c \triangleq \{ c_{ij}\supt \}_{i, j, \dtype \in \location}$ and $P\triangleq\{P_j\}_{j\in\location}$.

Based on rider demand $(\theta, \alpha)$ and the reported supply of drivers of each type, a  mechanism determines rider and driver flow, trip prices $p$, driver compensation $c$, and driver penalties $P$. Drivers decide whether or not to participate, considering pricing, penalties, and their outside options. Given entry decisions by drivers, and decisions made by drivers since
entry, the platform then dispatches drivers to trips and processes payments and penalties accordingly in each period.

\subsection{Steady-State Equilibrium}

In this section we will analyze a steady-state equilibrium while ignoring penalties. This will be helpful because it establishes that under truthful reporting of types, drivers will always follow the proposed dispatches. We will separately handle incentives to report truthful types, considering the effect of penalties on aligning these incentives. It will turn out that drivers are only charged penalties for the first time they deviate, not for future deviations.

At the beginning of each period, let $x_{i}\supt$ be the number of drivers of reported type $\dtype$ at location $i$,
and let $x\supt \triangleq \sum_{i \in \location} x_{i}\supt$ denote the total number of drivers of reported type $\dtype$ on the platform.
Let the \emph{trip flow} be $f \triangleq \{f_{ij}\supt\}_{i,j,\dtype\in \location}$, where $f_{ij}\supt \geq 0$ is the number of riders from $i$ to $j$ assigned to drivers of type $\dtype$. Let $y_{ij}\supt \geq 0$ be the mass of drivers of type $\dtype$ at $i$ who are dispatched to relocate to $j$ without a rider, and set $x \triangleq \{ x_{i} \supt \}_{i, \dtype \in \location}$ and $y \triangleq \{ y_{ij} \supt \} _{i,j,\dtype \in \location}$. No driver or rider can be matched multiple times in the same period, so assuming drivers always provide service, we have $\sum_{j \in \location} f_{ij}\supt +y_{ij}\supt\leq x_{i}\supt$ for all $i \in \location$ and all $\dtype \in \location$, and $\sum_{\dtype \in \location} f_{ij} \supt \leq \theta_{i}\alpha_{ij}(1-F(p_{ij}))$ and for all $i,j \in \location$. \vspace{0.3em}

For a trip with origin $i$ and destination $j$, if the total rider demand exceeds driver supply (i.e., if $\sum_{\dtype \in \location} f_{ij}\supt < \theta_{i}\alpha_{ij}(1-F(p_{ij}))$), the mechanism may increase the trip price $p_{ij}$ and achieve higher revenue. Therefore for revenue optimization, we can assume without loss that  $\sum_{\dtype \in \location} f_{ij}\supt= \theta_{i}\alpha_{ij}(1-F(p_{ij}))$. 
When $x_{i}\supt > 0$, meaning that some drivers with reported type~$\dtype$ are at location $i$, the probability that a given driver of reported type $\dtype$ is dispatched to destination $j$ is $(f_{ij}\supt+y_{ij}\supt)/x_i\supt$. 
Assuming a driver of type $\dtype$ has truthfully reported her type and will provide service in all periods, her lifetime expected utility for starting from location $i$ is of the form
\begin{align}
    \dutil_{i}\supt = & \sum_{j \in \location} (c_{ij}\supt + \delta\dutil_{j}\supt)  \frac{f_{ij}\supt+y_{ij}\supt}{x_i\supt} + I \cdot \one{i = \dtype}, \label{equ:driver_utility}
\end{align}
where $\one{\cdot}$ is the indicator function.
The first term in \eqref{equ:driver_utility} is the expected compensation and future utility a driver gets when dispatched to one of the $n$ possible destinations. The second term corresponds to the idiosyncratic utility  drivers get from starting trips in their favorite locations.

\begin{definition}[Steady-State Equilibrium]\label{defn:sse}
A \textit{steady-state equilibrium} under pricing policy  $(p, c)$
is a tuple $(f, x, y)$ s.t.:
\begin{enumerate}[({C}1)]
    \item (Driver best-response) Drivers providing service always maximizes their payoff, i.e.~$\forall i,\dtype \in \location$, $x_i\supt> 0 \im \forall k \in \location, ~\pi_i\supt \geq I\cdot\one{i = \dtype} + \delta \dutil_{k}\supt $. 
	\item (Flow balance) For all locations $i \in \location$ and  driver types $\dtype \in \location$, $x_{i} \supt = \sum_{j \in \location} f_{ji} \supt + y_{ji} \supt$.
	\item (Market-clearing) $\sum_{\dtype \in \location} f_{ij}\supt = \theta_i \alpha_{ij}(1-F(p_{ij}))$.
	\item (Individually rational driver entry) Participating drivers get at least their outside option $w$; with excess supply of drivers with type $\dtype$, all participating type-$\dtype$ drivers get exactly their outside option $w$. 
	\item (Feasibility) Rider and driver flows are non-negative, i.e., $\forall i,j,\dtype \in \location$, $f_{ij}\supt$, $y_{ij}\supt$, $x_{i}\supt\geq 0$; the supply constraints are satisfied, i.e., $\forall \dtype \in \location$, $\sum_{i \in \location} x_{i}\supt \leq s\supt$.
\end{enumerate}
\end{definition}

The {\em full information first best revenue} (FB) is the highest revenue a mechanism can achieve in stationary-state equilibrium, if the mechanism has full knowledge of driver types (therefore does not need to determine dispatching and compensation in order to incentivize truthful reporting of types):
\begin{align}
	\max_{p,c} &  \sum_{i \in \location} \sum_{j \in \location}\sum_{\dtype \in \location} p_{ij}\cdot f_{ij} \supt -c_{ij} \supt (f_{ij}\supt + y_{ij}\supt) \label{equ:rev_opt} \\
	\txtst & 
	(f,x,y) \text{ is a steady-state equilibrium under } (p,c). \notag 
\end{align}

The design problem is to compute rider prices $p$, driver compensation $c$, and driver penalties $P$ to optimize platform revenue in the steady state equilibrium, in a way that drivers will truthfully report their location preferences and will choose to always provide service.

\section{The Preference-Attentive Ridesharing Mechanism (PARM)}\label{sec:parm}

We now introduce our \emph{Preference-Attentive Ridesharing Mechanism} (\emph{PARM}) and show that this mechanism (i) truthfully elicits drivers' location preferences, (ii) incentivizes drivers to provide service, and (iii) achieves first-best revenue when supply is unconstrained or when demand is symmetric. 

\subsection{Alternate Form of the Optimization} 

The optimization problem \eqref{equ:rev_opt} need not be convex, and moreover, even when an optimal solution can be found, it may not incentivize drivers to report their types truthfully.
Denoting $W \triangleq w(1-\delta)$, we present an alternative problem \eqref{equ:rev_opt_rewritten}, which guarantees that any optimal solution can be converted into an optimal solution for \eqref{equ:rev_opt} using compensation scheme \eqref{equ:compensation}---while preserving the objective.
Specifically, we consider: 
\begin{align}
	\max_{p,f,x,y} 
	~& \Big( \hspace{-0.1em}  \sum_{i,j,\dtype \in \location} \hspace{-0.1em} f_{ij} \supt p_{ij}\Big) \hspace{-0.1em} - \hspace{-0.1em} W \sum_{i,\dtype \in \location }x_{i} \supt \hspace{-0.1em} + \hspace{-0.1em} I \sum_{\dtype \in \location} x_{\dtype} \supt  \label{equ:rev_opt_rewritten} \\
	\txtst &  x_{i} \supt  = \sum_{j \in \location} f_{ji} \supt +\sum_{j \in \location} y_{ji} \supt,~\forall i \in \location, ~\forall \dtype \in \location \notag \\
	 &  \sum_{\dtype \in \location} f_{ij} \supt = \theta_i \alpha_{ij}(1-F(p_{ij})), ~\forall i, j \in \location \notag  \\
	 &  \sum_{i \in \location}x_{i} \supt \leq  s \supt, ~\forall \dtype \in \location  \notag  \\
	 &  \sum_{j \in \location} y_{ij} \supt =  x_{i} \supt - \sum_{j \in \location} f_{ij} \supt, ~\forall i \in \location,~ \forall \dtype \in \location  \notag  \\
	 & f_{ij}\supt, y_{ij}\supt, x_{i}\supt  \geq 0, ~ \forall i,j \in \location,~ \forall \dtype \in \location.   \notag
\end{align} 
Our approach is analogous to a similar move by Bimpikis et al.~\shortcite{Bimpikis2016}--- assuming that $F$ is distributed $\mathrm{U}[0,1]$, the solution space is convex, and the optimization problem is quadratic. We also go a step further by accounting for driver heterogeneity and the possibility of zero demand at a location, the latter by paying drivers for relocation dispatches.

Consider the following compensation scheme: 
\begin{align}
    c_{ij} \supt \hspace{-0.2em} =
       \hspace{-0.2em} W \hspace{-0.2em} -I\cdot\one{i=\dtype}, \forall i,j, \dtype \in \location. \label{equ:compensation}
\end{align}

\begin{restatable}{lemma}{lemcformseq}
\label{lem:c_forms_eq}
Consider an optimal solution $(p,f,x,y)$ to problem \eqref{equ:rev_opt_rewritten}, and let $c$ be the compensation scheme \eqref{equ:compensation}. Then:
\begin{enumerate}[(i)]
    \itemsep 0em
    \item $(p,c)$ is an optimal solution to optimization problem \eqref{equ:rev_opt} with steady-state equilibrium $(f,x,y)$; and
    \item The expected lifetime utility (payment and location value) of a truthful
    driver is exactly $w$ starting from every location, i.e.~$\pi_i^{(\dtype)}=w$ for all $i,\dtype\in\location$.
\end{enumerate} %
\end{restatable}

Briefly, feasible solutions to \eqref{equ:rev_opt_rewritten} satisfy conditions (C2), (C3) and (C5). Moreover, with $W>I$, the 
compensation $c$ as in \eqref{equ:compensation} is non-negative.
Given \eqref{equ:compensation}, drivers receive utility $W$ in expectation per period (so $w$ over their lifetimes); this implies (C1) and (C4). Furthermore, the solution is optimal, since no compensation scheme can lower the total payment to drivers while fulfilling the same rider trip flow $f$.

\subsection{Constructing PARM}
\begin{definition} \label{defn:parm}
Given rider demand $(\theta, ~\alpha)$, the \textit{Preference-Attentive Ridesharing Mechanism} (PARM):
\begin{enumerate}[1.]
    \itemsep0em
    \item Elicits the location preferences from drivers.
    \item Solves \eqref{equ:rev_opt_rewritten} with an additional constraint
    \begin{align}
        x_{\dtype} \supt \geq x_{i} \supt, ~\forall i \in \location, ~ \forall \dtype \in \location \label{equ:ic_constraint}
    \end{align}
    for dispatching and pricing, and determine driver compensation by \eqref{equ:compensation}.
    \item If a driver with reported type-$\dtype$ did not provide service and relocated to location $i\neq\dtype \in \location$, the platform treats her as a type-$i$ driver from then on. If this is the first deviation for this driver, the driver pays penalty $P_\dtype\triangleq \max \{ \max_{k \in \location} \{ P^{k \ra \dtype} \},~ 0\}$ for $P^{k \ra \dtype}$ as solved for in the following linear system:
    \begin{align}
        \pi_i^{k\ra\dtype} \hspace{-0.2em} = & \one{i\neq\dtype} \left(W+\delta\sum_j \dfrac{f_{ij}\supt+y_{ij}\supt}{x_i \supt }\pi_j^{k\ra\dtype}\right) + \notag \\
        &  \one{i=\dtype}(\delta w \hspace{-0.1em} - \hspace{-0.1em} P^{k\ra\dtype}) \hspace{-0.1em} + \hspace{-0.1em} \one{i=k}I, ~\forall i,k \in \location; \notag \\
        w=&\sum_i \pi_i^{k\ra\dtype}  x_i\supt / x\supt
         , ~\forall k \in \location. \label{equ:misreport_total}
    \end{align}
\end{enumerate}
\end{definition}

The system \eqref{equ:misreport_total} has $n^2+n$ linear equations and $n^2+n$ unknowns ($n^2$ of the $\pi_i^{k\ra\dtype}$ and $n$ of the $P^{k\ra\dtype}$).
Intuitively, $\pi_i^{k \rightarrow \dtype}$ is the expected utility of a driver of type $k$ pretending to be of type $\dtype$ and providing service everywhere except $\dtype$, where she instead relocates to $k$. By construction, $P^{k \ra \dtype}$ is the minimum penalty needed to equalize driver earnings between this deviation and truth telling plus always providing service. We take the maximum over such penalties so no driver can benefit from pretending to be of type $\dtype$ and employ this strategy.
\footnote{
The penalty $P_\dtype$ is set to zero if $P^{k \ra \dtype} < 0$ for all $k \in \location$. This case arises if this deviation is itself bad for drivers of all types, in which case the only way to make the deviating drivers' utility equal to $w$ is to pay those drivers.}

If a driver declines to provide service but relocates to her reported preferred location, she is charged no penalty. A driver might have a legitimate (idiosyncratic) reason for not being able to provide service in a period, but if she relocates to a location she did not report as preferred, that is taken as an indication that her original report was not truthful.

We now prove, under the  assumption that drivers always provide service and as a result are never charged any penalty, that imposing \eqref{equ:ic_constraint} is sufficient to guarantee truthful reporting.

\begin{theorem} \label{thm:IC_assuming_no_deviation}
Assuming all drivers always provide service,
it is a dominant strategy for drivers to  report their location preferences truthfully under PARM. 
\end{theorem}

\begin{proof} 
Observe that by being truthful, each driver gets utility $W \hspace{-0.1em} = \hspace{-0.1em} w(1-\delta)$ per period---getting paid $W-I$ at  preferred locations, and $W$ at every other location. As a result, $\pi_i\supt \hspace{-0.1em}  = \hspace{-0.1em} w$ for all $i,\dtype \in \location$.
Suppose an infinitesimal driver of type $k \in \location$ reports she is of type $\dtype \neq k$. At all $i \neq k,\dtype$ she gains utility $W$ per period.
At $k$, she makes $W+I$ because the platform, treating her as a type-$\dtype$ driver, is still paying her $W$. 
At $\dtype$, she is paid $W-I$ and does not get the extra utility $I$.

With $\delta \rightarrow 1$, misreporting  $\dtype$ in place of $k$ leads to an increase in the expected payoff in static steady-state equilibrium if and only if in equilibrium, the driver with reported type $\dtype$ spends more time in location $k$ than in location $\dtype$. Considering the location of a driver with reported type $\dtype$ as a Markov chain, then $\{x_i\supt/x\supt\}_{i \in \location}$ is the stationary distribution. \eqref{equ:ic_constraint} then guarantees that a driver with reported type $\dtype$ spends a plurality of her time at location $\dtype$, therefore no driver benefits from misreporting her type if all drivers always provide service.
\end{proof}

We now consider drivers who may strategically decline to provide service and show such deviations are not useful under PARM, which updates its belief about a driver's type after deviations and imposes a penalty on the first such deviation.

\begin{restatable}{theorem}{thmparmic}\label{thm:parm_IC}
Under PARM, it is an ex post Nash equilibrium  for drivers to report their types truthfully and to always provide service.
\end{restatable}

Briefly, Theorem~\ref{thm:IC_assuming_no_deviation} and the following Lemma~\ref{lem:strat} imply that (i) a profitable misreport must be paired with post-reporting deviation(s), and (ii) the most profitable deviation must be the driver providing service everywhere except her reported preferred location. The penalties ensure the driver does not get a utility higher than $w$ from this deviation (or any other), so there does not exist a profitable deviation.

Drivers are never charged any penalty under the equilibrium outcome, but the threat of a penalty is necessary to ensure truthful reporting. In certain special economies, a misreporting driver might spend many periods at her true preferred location before being sent to her reported preferred location. Without penalties, she may simply decline service and relocate back to her actual preferred location, thereby sacrificing one period of income for the possibility of many periods of extra idiosyncratic utility. See Appendix~B of the full paper for an example and discussions.

\begin{restatable}{lemma}{lemstrat} \label{lem:strat}
Consider a driver of true type $k \in \location$ and reported type $\dtype \in \location$, and assume that the rest of the drivers always provide service. If $\dtype = k$ (truthful), always providing service is a best response. If $\dtype \neq k$, one of the following is a best-response: (i) always providing service, or (ii) providing service at every location except $\dtype$, where the driver instead drives to $k$.
\end{restatable}

We now outline the proof of Lemma~\ref{lem:strat}. We first show that a truthful driver earns $W$ at every location, and it is always optimal for her to provide service.%
An untruthful driver gets the least utility when at her reported preferred location  $\dtype$, so relocating to $\dtype$ is worse than providing service; in any period, she should either provide service or relocate to a location $i \neq \dtype$. In fact, because she is charged the same penalty for relocating to any location $i\neq\dtype$, her optimal relocation is her true preferred location $k$. Intuitively, if she relocates to $i$, she will be at $i$ and treated as type $i$ from then on, which is sub-optimal for her unless $i = k$. Given that the optimal relocation is then her true preferred location---after which she will make $W$ every period---the only location where she might profitably not provide service is $\dtype$, her reported preferred location and the only place she currently makes less than $W$ in-period.

\subsection{Cases with First-Best Revenue and No Penalty}

Although the IC constraint \eqref{equ:ic_constraint} may reduce revenue, we can characterize some settings where imposing IC does not lead to a revenue loss: when supply is unconstrained, or when rider demand is symmetric, PARM achieves the full-information first-best revenue. Furthermore, no penalty is necessary to ensure incentive compatibility.

\begin{restatable}{theorem}{thminfinitesupply}
\label{thm:inf_supply}
Suppose $s \supt =\infty$ for all $\dtype \in \location$. Then PARM achieves full-information first-best revenue, and no penalty is necessary to ensure incentive compatibility.
\end{restatable}

Briefly, the IC constraint \eqref{equ:ic_constraint} does not bind because drivers cost less to the platform when at their preferred locations. If there are more drivers with reported type $\dtype$ at location $i$ than at $\dtype$, the platform can improve its revenue by replacing the type-$\dtype$ drivers with type-$i$ drivers (there are unlimited type-$i$ drivers). Incentive compatibility holds without penalties because each driver visits her reported preferred location before visiting any other location too many times. Thus, she cannot profitably use a misreport-plus-deviation to sacrifice one period of income for many periods of idiosyncratic utility (as described following Theorem~\ref{thm:parm_IC}).
Note that the preceding argument makes no assumption on the demand pattern and requires only the availability of supply.

\begin{definition} \label{defn:symmetric_demand}  Rider demand $(\theta, \alpha)$ is \textit{symmetric} if $\forall i,j,k,l \in \location$ we have $\theta_i=\theta_k$ and $\alpha_{ij}=\alpha_{kl}$.
\end{definition}

\begin{restatable}{theorem}{thmfbsymmdemand}
\label{thm:fb_symmetric_demand} 
Suppose that rider demand is symmetric. Then we can construct a solution to optimization problem \eqref{equ:rev_opt_rewritten} with incentive compatibility  constraint \eqref{equ:ic_constraint} such that PARM achieves full-information first-best revenue, and no penalty is necessary to ensure incentive compatibility.
\end{restatable}

\noindent{}To understand Theorem~\ref{thm:fb_symmetric_demand}, we prove two additional lemmas.

\begin{restatable}{lemma}{lemmaselffill}
\label{lem:self_fill}
With symmetric demand, any optimal solution to \eqref{equ:rev_opt_rewritten} satisfies $f_{ii}\supt + y_{ii}\supt\leq f_{\dtype\dtype}\supt+y_{\dtype\dtype}\supt$ for all $i,\dtype \in \location$.
\end{restatable}

Intuitively, type-$\dtype$ drivers cost less when at location $\dtype$, so it is optimal for the marginal ride they give at location $\dtype$ to have a lower price than at other locations.
If demand is symmetric, this means  drivers with reported type-$\dtype$ provide more rides at location $\dtype$ than at any other location.

\begin{restatable}{lemma}{lemmabilateral}
\label{lem:bilateral}
If the demand pattern is symmetric, we can construct an optimal solution to \eqref{equ:rev_opt_rewritten} such that for all $i,j\in \location$ and all $\dtype \in \location$,  $f_{ij}\supt=f_{ji}\supt$ and $y_{ij}\supt=y_{ji}\supt=0$.
\end{restatable}
Intuitively, there is no need for drivers to relocate when demand is fully symmetric. Moreover, given any optimal solution to \eqref{equ:rev_opt_rewritten}, we can construct an alternative optimal solution, where the flow of drivers of each type can be decomposed as cycles with length $2$, i.e., $f_{ij}\supt=f_{ji}\supt$.

\smallskip

We can now sketch the proof of Theorem~\ref{thm:fb_symmetric_demand}. 
With symmetric demand, Lemma~\ref{lem:self_fill} implies driver flow for within-location trips satisfies the IC constraint \eqref{equ:ic_constraint}. For all inter-location trips, Lemma~\ref{lem:bilateral} lets us focus only on bilateral driver flow between pairs of locations $i$ and $j$. Type-$\dtype$ drivers cost less at $\dtype$, so they will naturally fill more rides between $\dtype$ and $j$ than between $i$ and $j$---and this holds for all $j$, so type-$\dtype$ drivers fill more rides in and out of $\dtype$ than $i$. Combining the two cases, drivers of type $\dtype$ do not spend more time at another location $i\neq \dtype$ than they do at $\dtype$, so imposing the IC constraint \eqref{equ:ic_constraint} is without loss of revenue. As in Theorem~\ref{thm:inf_supply}, incentive compatibility holds without penalties because each driver will visit her reported preferred location before visiting any other too many times. Thus, she cannot profitably use a misreport-plus-deviation to sacrifice one period of income for many periods of idiosyncratic utility (as described following Theorem~\ref{thm:parm_IC}). 
\section{Simulation Results} \label{sec:simulations}
In this section, we use simulations to analyze the revenue and social welfare under PARM for settings outside the cases covered by Theorems~\ref{thm:inf_supply} and~\ref{thm:fb_symmetric_demand}---i.e., settings with limited  supply and unbalanced demand. 

Social welfare is defined as the total rider value plus drivers' utilities from being in their preferred locations, minus the total opportunity costs incurred by drivers.
We compare PARM with the full-information {\em first-best}, and also a {\em Preference-Oblivious Ridesharing Mechanism} (\emph{PORM}) which sets prices as in \citet{Bimpikis2016} without considering drivers' location preferences, while assuming that drivers always follow  dispatches. 
In Section~\ref{sec:simulation_equilibrium}, we also study the equilibrium outcome under PORM, allowing driver autonomy.
For ease of illustration, we consider two locations $\location = \{0, 1\}$ throughout the analysis.

\newcommand{\subfigSize}{0.235}

\subsection{Varying Demand Patterns} 

Suppose that there are an equal number of drivers favoring each location: $s^{(0)} \hspace{-0.2em} = \hspace{-0.2em}  s^{(1)} \hspace{-0.2em}  = \hspace{-0.2em} 100$. Drivers have outside option $w \hspace{-0.2em} = \hspace{-0.2em} 40$, discount factor $\delta \hspace{-0.1em} = \hspace{-0.1em} 0.99$, and gain utility $I =\hspace{-0.1em} 0.2W \hspace{-0.1em} = 0.2w(1-\delta)$ per period from being in their preferred locations. Each rider has  value independently drawn  $\sim \mathrm{U}[0,1]$. 
\paragraph{Varying Total Demand.}
\newcommand{\subfigSizeedit}{\footnotesize} 

We first assume an unbalanced trip flow $\alpha_{00} \hspace{-0.2em}= \hspace{-0.2em}\alpha_{10}=0.25$ and $\alpha_{01} \hspace{-0.2em}= \hspace{-0.2em}\alpha_{11}=0.75$ (i.e., three quarters of riders from each location would like to go to location $1$). Fixing the total demand at location $1$ at $\theta_1\hspace{-0.1em} = \hspace{-0.1em}1000$, and varying $\theta_0$ from $0$ to $1000$, the revenue and welfare under PARM and benchmarks are as in Figure~\ref{fig:varying_theta}.
Although PARM only necessarily achieves first-best revenue when $\theta_0 \hspace{-0.1em}= \hspace{-0.1em}1000$ (symmetric demand), we see that PARM achieves the first-best and outperforms PORM unless $\theta_0$ is very small, such that demand from the two locations is highly asymmetric.
\begin{figure}[t!]
    \centering
    \begin{subfigure}[b]{\subfigSize\textwidth}
        \centering
        \includegraphics[width=\textwidth]{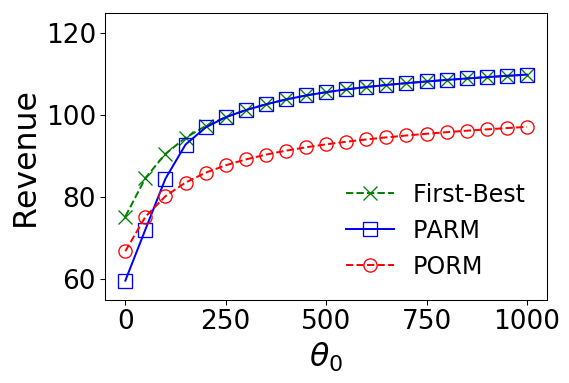}
        \vspace{-1.5\baselineskip}
        \caption[eqm_rev_i]%
        {{\small Revenue}} 
        \label{subfig_rev_theta0}
    \end{subfigure}
    ~
    \begin{subfigure}[b]{\subfigSize\textwidth}  
        \centering 
        \includegraphics[width=\textwidth]{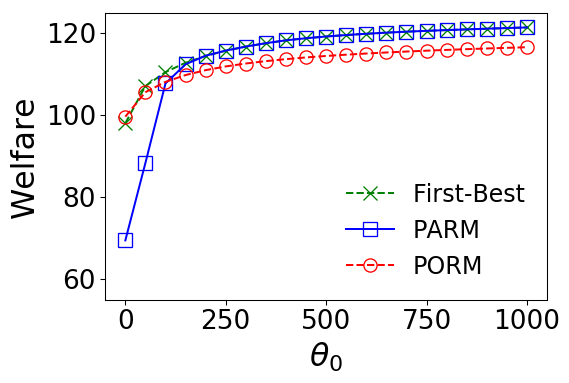}
        \vspace{-1.5\baselineskip}
        \caption[eqm_welf_i]%
        {{\small Welfare}}    
        \label{subfig_welfare_theta0}
    \end{subfigure}
    \vspace{-1.5\baselineskip}
    \caption[]%
    {{ Revenue and welfare varying demand $\theta_0$ at location~0.}}
    \label{fig:varying_theta}
\end{figure}

When $\theta_0 \ll \theta_1$, almost all rides originate and terminate at location~$1$, thus the first-best and  PORM dispatch most drivers of both types to provide service at location~$1$.  Figure~\ref{fig:vary_theta_snap_theta_50} illustrates the rider trip flows fulfilled by drivers of each type under different mechanisms, when $\theta_0 = 50$.
To satisfy PARM's incentive compatibility (IC) constraint, however, drivers of type $0$ must spend a plurality of their time at location 0. Therefore, PARM completes fewer trips at location~$1$, dispatches more type~$0$ drivers to fulfill (the less profitable) between-location trips, and asks many type~$0$ drivers to relocate back to~$1$ once they arrive at location $0$ (the numbers after the ``+" sign represent driver relocation flow), resulting in lower revenue and social welfare. 

\newcommand{\flowFigWidthMult}{0.15}
\newcommand{\flowFigFontSize}{\footnotesize} 
\newcommand{\flowFigLabelAbove}{5pt}
\newcommand{\flowFigCrlSize}{7mm}
\newcommand{\flowFigArcParam}{-45:225:3.5mm}
\newcommand{\flowFigCircleGap}{20pt}
\newcommand{\flowFigArrowSize}{3pt}

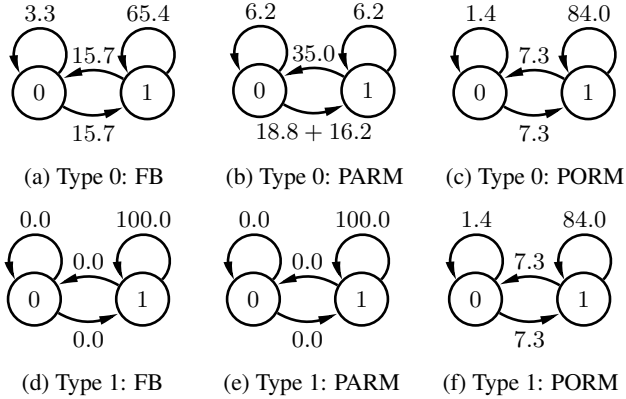
\begin{figure}[t!]
    \centering
    \begin{subfigure}[b]{\flowFigWidthMult \textwidth}
        \centering
        \begin{tikzpicture}
            [ font = \flowFigFontSize,
              >/.tip={Triangle[length=7.5pt,width=\flowFigArrowSize,bend]},
              line width=1pt,
              my circle/.style={minimum width=\flowFigCrlSize, circle, draw},
              my label/.style={above=\flowFigLabelAbove, anchor=mid}
            ]
            \node (0) [my circle] {$0$};
            \node (1) [my circle, right= \flowFigCircleGap of 0] {$1$};
            
            \path (0) [->, draw] edge[bend right]  node[below]{$15.7$}   (1);
            \path (1) [->, draw] edge[bend right]  node[my label]{$15.7$}   (0);
            \path [->, draw] (1.north east) arc (\flowFigArcParam) node [midway, my label] {$65.4$};
            \path [->, draw] (0.north east) arc (\flowFigArcParam) node [midway, my label] {$3.3$};
        \end{tikzpicture}
        \caption[Network2]%
        {{\small Type 0: FB}} 
        \label{fig:vary_theta_snap_theta_50_type0_fb}
    \end{subfigure}
    ~
    \begin{subfigure}[b]{\flowFigWidthMult \textwidth}  
        \centering 
        \begin{tikzpicture}
            [ font = \flowFigFontSize,
              >/.tip={Triangle[length=7.5pt,width=\flowFigArrowSize,bend]},
              line width=1pt,
              my circle/.style={minimum width=\flowFigCrlSize, circle, draw},
              my label/.style={above= \flowFigLabelAbove, anchor=mid}
            ]
            \node (0) [my circle] {$0$};
            \node (1) [my circle, right= \flowFigCircleGap of 0] {$1$};
            
            \path (0) [->, draw] edge[bend right]  node[below]{$18.8+16.2$}   (1);
            \path (1) [->, draw] edge[bend right]  node[my label]{$35.0$}   (0);
            \path [->, draw] (1.north east) arc (\flowFigArcParam) node [midway, my label] {$6.2$};
            \path [->, draw] (0.north east) arc (\flowFigArcParam) node [midway, my label] {$6.2$};
        \end{tikzpicture}
        \caption[]%
        {{\small Type 0: PARM}}    
        \label{fig:vary_theta_snap_theta_50_type0_parm}
    \end{subfigure}
    ~
    \begin{subfigure}[b]{\flowFigWidthMult \textwidth}
        \centering
        \begin{tikzpicture}
            [font = \flowFigFontSize,
              >/.tip={Triangle[length=7.5pt,width=3pt,bend]},
              line width=1pt,
              my circle/.style={minimum width= \flowFigCrlSize, circle, draw},
              my label/.style={above= \flowFigLabelAbove, anchor=mid}
            ]
            \node (0) [my circle] {$0$};
            \node (1) [my circle, right= \flowFigCircleGap of 0] {$1$};
            
            \path (0) [->, draw] edge[bend right]  node[below]{$7.3$}   (1);
            \path (1) [->, draw] edge[bend right]  node[my label]{$7.3$}   (0);
            \path [->, draw] (1.north east) arc (\flowFigArcParam) node [midway, my label] {$84.0$};
            \path [->, draw] (0.north east) arc (\flowFigArcParam) node [midway, my label] {$1.4$};
        \end{tikzpicture}
        \caption[Network2]%
        {{\small Type 0: PORM}}    
        \label{fig:vary_theta_snap_theta_50_type0_porm}
    \end{subfigure}
    \vskip0.3\baselineskip
    \begin{subfigure}[b]{\flowFigWidthMult \textwidth}   
        \centering 
        \begin{tikzpicture}
            [ font = \flowFigFontSize,
              >/.tip={Triangle[length=7.5pt,width=\flowFigArrowSize,bend]},
              line width=1pt,
              my circle/.style={minimum width=\flowFigCrlSize, circle, draw},
              my label/.style={above = \flowFigLabelAbove, anchor=mid}
            ]
            \node (0) [my circle] {$0$};
            \node (1) [my circle, right= \flowFigCircleGap of 0] {$1$};
            
            \path (0) [->, draw] edge[bend right]  node[below]{$0.0$}   (1);
            \path (1) [->, draw] edge[bend right]  node[my label]{$0.0$}   (0);
            \path [->, draw] (1.north east) arc (\flowFigArcParam) node [midway, my label] {$100.0$};
            \path [->, draw] (0.north east) arc (\flowFigArcParam) node [midway, my label] {$0.0$}; 
        \end{tikzpicture}
        \caption[]%
        {{\small Type 1: FB}}
        \label{fig:vary_theta_snap_theta_50_type1_fb}
    \end{subfigure}
    ~
    \begin{subfigure}[b]{ \flowFigWidthMult \textwidth}   
        \centering 
        \begin{tikzpicture}
            [ font = \flowFigFontSize,
              >/.tip={Triangle[length=7.5pt,width=\flowFigArrowSize,bend]},
              line width=1pt,
              my circle/.style={minimum width=\flowFigCrlSize, circle, draw},
              my label/.style={above= \flowFigLabelAbove, anchor=mid}
            ]
            \node (0) [my circle] {$0$};
            \node (1) [my circle, right= \flowFigCircleGap of 0] {$1$};
            
            \path (0) [->, draw] edge[bend right]  node[below]{$0.0$}   (1);
            \path (1) [->, draw] edge[bend right]  node[my label]{$0.0$}   (0);
            \path [->, draw] (1.north east) arc (\flowFigArcParam) node [midway, my label] {$100.0$};
            \path [->, draw] (0.north east) arc (\flowFigArcParam) node [midway, my label] {$0.0$}; 
        \end{tikzpicture}
        \caption[]%
        {{\small Type 1: PARM}}    
        \label{fig:vary_theta_snap_theta_50_type1_parm}
    \end{subfigure}
    ~
    \begin{subfigure}[b]{\flowFigWidthMult \textwidth}   
        \centering 
        \begin{tikzpicture}
            [ font = \flowFigFontSize,
              >/.tip={Triangle[length=7.5pt,width=\flowFigArrowSize,bend]},
              line width=1pt,
              my circle/.style={minimum width = \flowFigCrlSize, circle, draw},
              my label/.style={above = \flowFigLabelAbove, anchor=mid}
            ]
            \node (0) [my circle] {$0$};
            \node (1) [my circle, right= \flowFigCircleGap of 0] {$1$};
            
            \path (0) [->, draw] edge[bend right]  node[below]{$7.3$}   (1);
            \path (1) [->, draw] edge[bend right]  node[my label]{$7.3$}   (0);
            \path [->, draw] (1.north east) arc (\flowFigArcParam) node [midway, my label] {$84.0$};
            \path [->, draw] (0.north east) arc (\flowFigArcParam) node [midway, my label] {$1.4$};
        \end{tikzpicture}
        \caption[]%
        {{\small Type 1: PORM}}    
        \label{fig:vary_theta_snap_theta_50_type1_porm}
    \end{subfigure}   
    \vspace{-0.5em}
    \caption[Flow]
    {Total rider trips fulfilled by drivers of each type,
    with $(\alpha_{i0}, \alpha_{i1})=(0.25,~0.75)$, $\theta=(50,~1000)$, and $s=(100,100)$.
    } 
    \label{fig:vary_theta_snap_theta_50}
\end{figure}

\paragraph{Varying Imbalance in Demand.} 

\begin{figure}[t!]
    \centering
    \begin{subfigure}[b]{\subfigSize\textwidth}
        \centering
        \includegraphics[width=\textwidth]{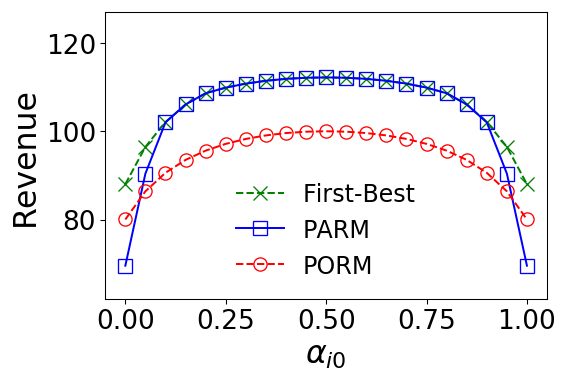}
        \vspace{-1.5\baselineskip}
        \caption[eqm_rev_i]%
        {{\small Revenue}} 
        \label{fig:varying_alpha_revenue}
    \end{subfigure}
    ~
    \begin{subfigure}[b]{\subfigSize\textwidth}  
        \centering 
        \includegraphics[width=\textwidth]{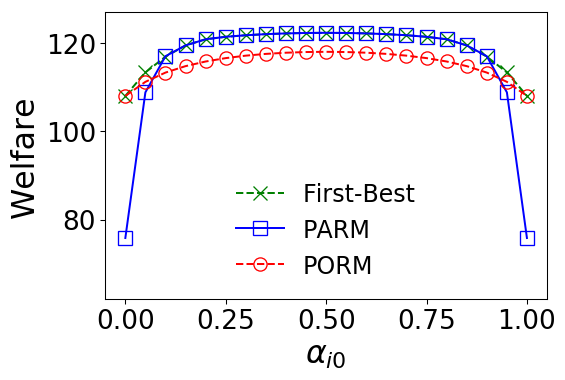}
        \vspace{-1.5\baselineskip}
        \caption[eqm_welf_i]%
        {{\small Welfare}}    
        \label{fig:varying_alpha_welfare}
    \end{subfigure}
    \vspace{-1.5\baselineskip}
    \caption[]%
    {Revenue and welfare varying $\alpha_{i0}$ for $i=0,1$. \label{fig:varying_alpha}} 
\end{figure}

Fixing $\theta_0=\theta_1=1000$ and varying $\alpha_{i0}$ for $i =0,1$ (i.e., changing the proportion of rides with destination~$0$), the revenue and welfare achieved by different mechanisms are shown in Figure~\ref{fig:varying_alpha}.  
Similar to Figure~\ref{fig:varying_theta}, PARM achieves first-best revenue and outperforms PORM for a wide range of $\alpha_{i0}$ (although demand is only symmetric when $\alpha_{i0} = 0.5$).
For similar reasons as in the above scenario, we see a decline of revenue and welfare under PARM when demand becomes highly unbalanced---in this case, when $\alpha_{i0}$ approaches $0$ or $1$ and almost all riders have the same destination.

\subsection{PORM in Equilibrium} \label{sec:simulation_equilibrium}

In this section, we analyze a scenario for which we are able
to compute the equilibrium outcome given the pricing under
PORM, and under the setting where drivers are given the flexibility to decide how to drive.
Consider two locations $\location \hspace{-0.1em}=\hspace{-0.1em} \{0,1\}$ and drivers of type~$1$ only: $s^{(0)}\hspace{-0.1em}=\hspace{-0.1em}0,~s^{(1)} \hspace{-0.1em}=\hspace{-0.1em}200$. All trips start and end in the same location, i.e., $\alpha_{00} \hspace{-0.1em}=\hspace{-0.1em} \alpha_{11}\hspace{-0.1em} = \hspace{-0.1em}1$. 
Being oblivious to drivers' preferences, PORM sets the same trip price for the two locations and expects the spatial distribution
of drivers to be proportional to the distribution of demand. In equilibrium, however, more drivers decide to drive in location~$1$ (the preferred location), such that in each period drivers in~$1$ are dispatched with probability less than $1$ and achieve the same expected utility as drivers in~$0$.

\paragraph{Varying Location Preference $I$.}
In Figure~\ref{fig:varying_I_eqm}, we fix demand $\theta_0=\theta_1 = 1000$ and plot revenue and welfare as $I$, the idiosyncratic driver utility, varies from $0$ to $W$. 
As $I$ increases, welfare and revenue under PARM coincide with the first-best and increase as expected. However, revenue under PORM (assuming driver compliance) remains constant since the mechanism is oblivious to drivers' preferences.
We also see a decrease in welfare and revenue achieved in equilibrium under PORM, since more drivers decide to supply in location~$1$, instead of in location~$0$ as dispatched, resulting in unfulfilled rides in~0 and idle drivers in~1. Beyond $I=0.5W$, revenue and welfare remain constant, since all drivers are already supplying location~$1$.

\begin{figure}[t!]
    \centering
    \begin{subfigure}[b]{\subfigSize\textwidth}
        \centering
        \includegraphics[width=\textwidth]{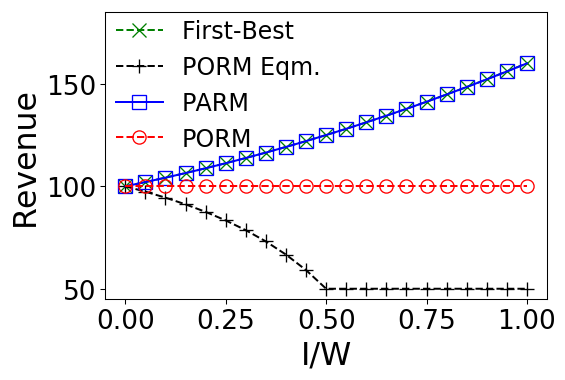}
        \vspace{-1.5\baselineskip}
        \caption[eqm_rev_i]%
        {{\small Revenue}} 
        \label{fig:varying_I_eqm_revenue}
    \end{subfigure}
    ~
    \begin{subfigure}[b]{\subfigSize\textwidth}  
        \centering 
        \includegraphics[width=\textwidth]{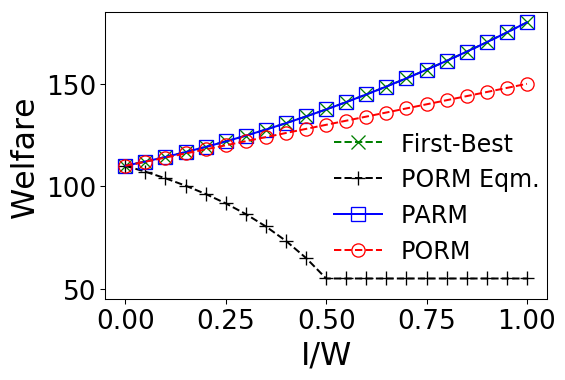}
        \vspace{-1.5\baselineskip}
        \caption[eqm_welf_i]%
        {{\small Welfare}}    
        \label{fig:varying_I_eqm_welfare}
    \end{subfigure}
    \vspace{-1.5\baselineskip}
    \caption[]%
    {Equilibrium revenue and welfare varying $I/W$.} 
    \label{fig:varying_I_eqm}
\end{figure}

Figure~\ref{fig:varying_I_eqm_flow} illustrates rider trip flow fulfilled by the type~$1$ drivers when $I = 0.2W$. PARM assigns more drivers to location 1 than location 0, but PORM does not. However, in equilibrium more drivers end up at location 1 anyway, leading to 25 units of drivers idling at location~$1$.

\newcommand{\flowFigSideWidth}{0.135}
\newcommand{\flowFigMidWidth}{0.18}
\renewcommand{\flowFigCircleGap}{18pt}
\renewcommand{\flowFigCrlSize}{6.5mm}   

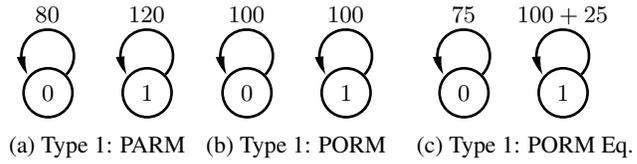
\begin{figure}[t!]
    \centering
    \begin{subfigure}[b]{ \flowFigSideWidth \textwidth}  
        \centering 
        \begin{tikzpicture}
            [ font = \flowFigFontSize,
              >/.tip={Triangle[length=7.5pt,width=\flowFigArrowSize,bend]},
              line width=1pt,
              my circle/.style={minimum width=\flowFigCrlSize, circle, draw},
              my label/.style={above=\flowFigLabelAbove, anchor=mid}
            ]
            \node (0) [my circle] {$0$};
            \node (1) [my circle, right=\flowFigCircleGap of 0] {$1$};
            \path [->, draw] (1.north east) arc (\flowFigArcParam) node [midway, my label] {$120$};
            \path [->, draw] (0.north east) arc (\flowFigArcParam) node [midway, my label] {$80$};
        \end{tikzpicture}
        \caption[]%
        {{\small Type 1: PARM}}    
        \label{fig:varying_I_eqm_flow_parm}
    \end{subfigure}
    ~    
    \begin{subfigure}[b]{ \flowFigSideWidth \textwidth}
        \centering
        \begin{tikzpicture}
            [ font = \flowFigFontSize,
              >/.tip={Triangle[length=7.5pt,width=\flowFigArrowSize,bend]},
              line width=1pt,
              my circle/.style={minimum width= \flowFigCrlSize, circle, draw},
              my label/.style={above=\flowFigLabelAbove, anchor=mid}
            ]
            \node (0) [my circle] {$0$};
            \node (1) [my circle, right=\flowFigCircleGap of 0] {$1$};
            \path [->, draw] (1.north east) arc (\flowFigArcParam) node [midway, my label] {$100$};
            \path [->, draw] (0.north east) arc (\flowFigArcParam) node [midway, my label] {$100$};
        \end{tikzpicture}
        \caption[Network2]%
        {{\small Type 1: PORM}}    
        \label{fig:varying_I_eqm_flow_porm}
    \end{subfigure}
    ~
    \begin{subfigure}[b]{ \flowFigMidWidth \textwidth}   
        \centering 
        \begin{tikzpicture}
            [ font = \flowFigFontSize,
              >/.tip={Triangle[length=7.5pt,width=\flowFigArrowSize,bend]},
              line width=1pt,
              my circle/.style={minimum width=\flowFigCrlSize, circle, draw},
              my label/.style={above=\flowFigLabelAbove, anchor=mid}
            ]
            \node (0) [my circle] {$0$};
            \node (1) [my circle, right=\flowFigCircleGap of 0] {$1$};
            \path [->, draw] (1.north east) arc (\flowFigArcParam) node [midway, my label] {$100+25$};
            \path [->, draw] (0.north east) arc (\flowFigArcParam) node [midway, my label] {$75$};
        \end{tikzpicture}
        \caption[]%
        {{\small Type 1: PORM Eq.}}    
        \label{fig:varying_I_eqm_flow_porm_equ}
    \end{subfigure}
    \vspace{-0.5\baselineskip}
    \caption[Flow]
    {Rider trips fulfilled by type-$1$ drivers, with $s=(0,200)$, $\theta=(1000,1000)$, $\alpha_{00} \hspace{-0.1em} = \hspace{-0.1em} \alpha_{11}  \hspace{-0.1em}=  \hspace{-0.1em}1$, $\alpha_{01}  \hspace{-0.1em}=  \hspace{-0.1em}\alpha_{10} \hspace{-0.1em} =  \hspace{-0.1em}0$,  
     and $I=0.2W$. \label{fig:varying_I_eqm_flow} }
\end{figure}

\paragraph{Varying Demand Ratio $\theta_0/\theta_1$.} 

In Figure~\ref{fig_eqm_theta}, we fix $\theta_1=1000$, $I/W=0.2$, and vary $\theta_0$ from $0$ to $2000$. We see that PARM revenue coincides with the first-best and significantly exceeds the revenue of PORM. The revenue and welfare of the equilibrium outcome under PORM is much lower, however, because drivers over-supply the preferred location~$1$, leaving rider trips in $0$ unfulfilled. It is curious that with highly unbalanced demand, an increase in $\theta_0$ initially leads to reduced equilibrium revenue and welfare---this is because with higher demand at location~$0$, PORM sets a higher price at location~$1$ and accepts fewer location~$1$ trips in order to complete more trips in~$0$. The drivers, however, are only willing to drive in $0$ when $\theta_0$ is high enough that the low probability of getting a ride in $1$ offsets the extra utility $I$.

\begin{figure}[t!]
    \centering
    \begin{subfigure}[b]{\subfigSize\textwidth}
        \centering
        \includegraphics[width=\textwidth]{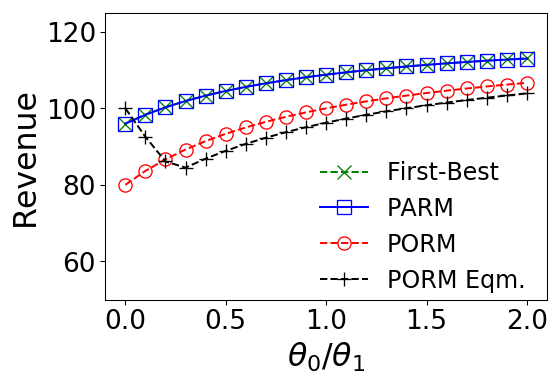}
        \vspace{-1.5\baselineskip}
        \caption[eqm_rev_i]%
        {{\small Revenue}} 
        \label{subfig_eqm_rev_t10_second}
    \end{subfigure}
    ~
    \begin{subfigure}[b]{\subfigSize\textwidth}  
        \centering 
        \includegraphics[width=\textwidth]{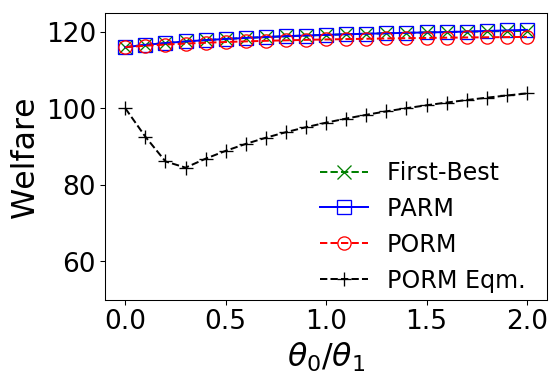}
        \vspace{-1.5\baselineskip}
        \caption[eqm_welf_i]%
        {{\small Welfare}}    
        \label{subfig_eqm_welfare_t10_second}
    \end{subfigure}
    \vspace{-1.5\baselineskip}
    \caption[]%
    {{ Equilibrium revenue and welfare varying $\theta_1/\theta_0$}} 
    \label{fig_eqm_theta}
\end{figure}

\paragraph{Varying Demand Ratio $\theta_1/\theta_0$.}  

In Figure~\ref{fig_eqm_t10}, we set $I/W=0.2$, $\theta_0 = 1000$ and vary $\theta_1$ from 0 to 2000. We see a similar trend as in Figure~\ref{fig:varying_theta}, with PARM doing worse than even equilibrium PORM for very small values of $\theta_1$. 
Figure~\ref{fig_flow_i} illustrates driver flow for $\theta_1 = 100 = 0.1\theta_0$--- PARM employs drivers to idle at location~1 in order to satisfy IC, and this is very costly.
It is worth noting that as $\theta_1$ increases, the social welfare achieved by the equilibrium outcome under PORM in fact does not increase, due to the increased amount of idle drivers at location~$1$.

\begin{figure}[htpb!]
    \centering
    \begin{subfigure}[b]{\subfigSize\textwidth}
        \centering
        \includegraphics[width=\textwidth]{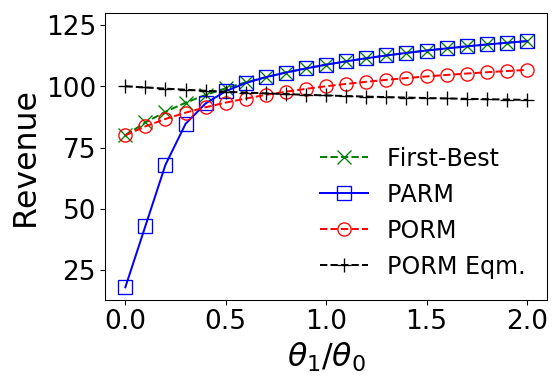}
        \vspace{-1.5\baselineskip}
        \caption[eqm_rev_i]%
        {{\small Revenue}} 
        \label{subfig_eqm_rev_t10_first}
    \end{subfigure}
    ~
    \begin{subfigure}[b]{\subfigSize\textwidth}  
        \centering 
        \includegraphics[width=\textwidth]{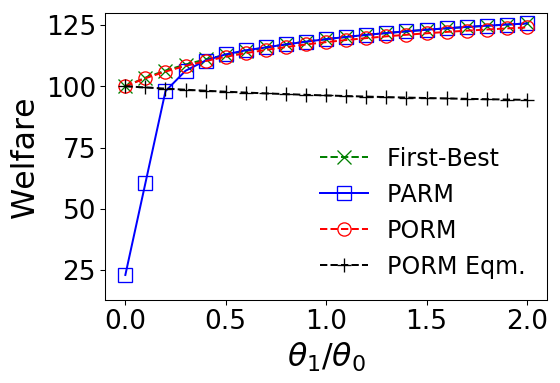}
        \vspace{-1.5\baselineskip}
        \caption[eqm_welf_i]%
        {{\small Welfare}}    
        \label{subfig_eqm_welfare_t10_first}
    \end{subfigure}
    \vspace{-1.5\baselineskip}
    \caption[]%
    {{ Equilibrium revenue and welfare varying $\theta_1/\theta_0$}.} 
    \label{fig_eqm_t10}
\end{figure}
\vspace{-3ex}
\begin{figure}[htpb!]
    \centering
    \begin{subfigure}[b]{ \flowFigSideWidth \textwidth}
        \centering
        \begin{tikzpicture}
            [ font = \flowFigFontSize,
              >/.tip={Triangle[length=7.5pt,width=\flowFigArrowSize,bend]},
              line width=1pt,
              my circle/.style={minimum width=\flowFigCrlSize, circle, draw},
              my label/.style={above=\flowFigLabelAbove, anchor=mid}
            ]
            \node (0) [my circle] {$0$};
            \node (1) [my circle, right=\flowFigCircleGap of 0] {$1$};
            \path [->, draw] (1.north east) arc (\flowFigArcParam) node [midway, my label] {$18.2$};
            \path [->, draw] (0.north east) arc (\flowFigArcParam) node [midway, my label] {$181.8$};
        \end{tikzpicture}
        \caption[Network2]%
        {{\small Type 1: PORM}}    
        \label{fig:mean and std of net14_second}
    \end{subfigure}
    ~
    \begin{subfigure}[b]{ \flowFigMidWidth \textwidth}   
        \centering 
        \begin{tikzpicture}
            [ font = \flowFigFontSize,
              >/.tip={Triangle[length=7.5pt,width=\flowFigArrowSize,bend]},
              line width=1pt,
              my circle/.style={minimum width=\flowFigCrlSize, circle, draw},
              my label/.style={above=\flowFigLabelAbove, anchor=mid}
            ]
            \node (0) [my circle] {$0$};
            \node (1) [my circle, right=\flowFigCircleGap of 0] {$1$};
            \path [->, draw] (1.north east) arc (\flowFigArcParam) node [midway, my label] {$18.2+4.5$};
            \path [->, draw] (0.north east) arc (\flowFigArcParam) node [midway, my label] {$177.3$};
        \end{tikzpicture}
        \caption[]%
        {{\small Type 1: PORM Eq.}}    
        \label{fig:mean and std of net34_first}
    \end{subfigure}
    ~
    \begin{subfigure}[b]{ \flowFigSideWidth \textwidth}  
        \centering 
        \begin{tikzpicture}
            [ font = \flowFigFontSize,
              >/.tip={Triangle[length=7.5pt,width=\flowFigArrowSize,bend]},
              line width=1pt,
              my circle/.style={minimum width=\flowFigCrlSize, circle, draw},
              my label/.style={above=\flowFigLabelAbove, anchor=mid}
            ]
            \node (0) [my circle] {$0$};
            \node (1) [my circle, right=\flowFigCircleGap of 0] {$1$};
            \path [->, draw] (1.north east) arc (\flowFigArcParam) node [midway, my label] {$50.0+50.0$};
            \path [->, draw] (0.north east) arc (\flowFigArcParam) node [midway, my label] {$100.0$};
        \end{tikzpicture}
        \caption[]%
        {{\small Type 1: PARM}}    
        \label{fig:mean and std of net24}
    \end{subfigure}
    \caption[Flow]
    {\small Rider trip and idle driver flows, with $\theta=(1000,100)$, $\alpha=[(1.0,0.0),(0.0,1.0)]$, $s=(0,200)$, and $I=0.2W$.}
    \label{fig_flow_i}
\end{figure}
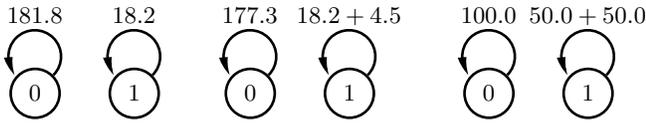

\paragraph{Varying Driver Supply $s_1$.} We now examine the effect of varying the supply of type $1$ drivers (while still keeping the supply of type~0 drivers at $0$). In Figure~\ref{fig_eqm_s}, we set $I=0.2W$, $\theta_0=\theta_1=1000$, and vary $s_1$ from $0$ to $1000$. Revenue and welfare under PARM coincide with first-best and outperform PORM. All the mechanisms improve in profit and welfare as supply increases, but PARM is better able to use the additional drivers. Under PORM in equilibrium, drivers again over-supply the preferred location~$1$, causing rides at location $0$ to get dropped. Eventually, there are so many drivers that they can fill all the demand, even with drivers idling at location $1.$ At this point, equilibrium PORM revenue coincides with PORM revenue, though the welfare is still lower.

\begin{figure}[htpb!]
    \centering
    \begin{subfigure}[b]{\subfigSize\textwidth}
        \centering
        \includegraphics[width=\textwidth]{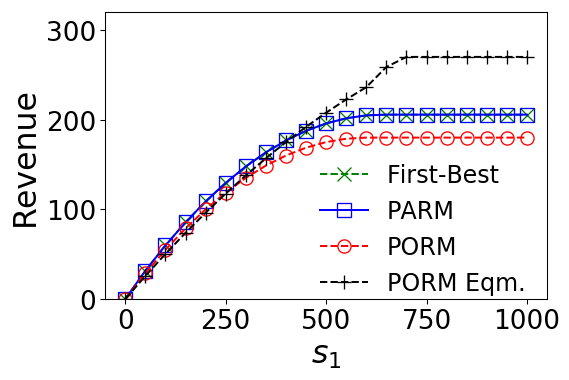}
        \vspace{-1.5\baselineskip}
        \caption[eqm_rev_i]%
        {{\small Revenue}} 
        \label{subfig_eqm_rev_s}
    \end{subfigure}
    ~
    \begin{subfigure}[b]{\subfigSize\textwidth} 
        \centering 
        \includegraphics[width=\textwidth]{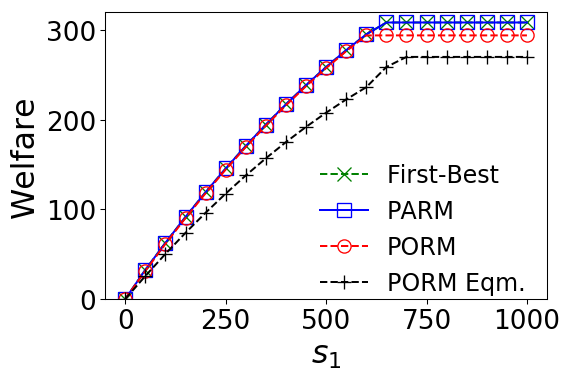}
        \vspace{-1.5\baselineskip}
        \caption[eqm_welf_i]%
        {{\small Welfare}}    
        \label{subfig_eqm_welfare_s}
    \end{subfigure}
    \vspace{-1.5\baselineskip}
    \caption[]
    {{ Equilibrium revenue and welfare varying $s_1$}} 
    \label{fig_eqm_s}
\end{figure}


\section{Discussion} \label{sec:conclusion}

We have proposed the Preference-Attentive Ridesharing Mechanism (PARM) for pricing and dispatch in the presence of driver location preferences. It is an equilibrium under PARM for drivers to  report their preferred locations truthfully and always provide service. PARM achieves first-best revenue in settings with unconstrained driver supply or symmetric rider demand, and we show via simulations that even outside those scenarios, PARM achieves close to first-best welfare and revenue and outperforms a mechanism that is oblivious to location preferences.

Our analysis suggests that incorporating drivers' location preferences is compatible with other aspects of ridesharing pricing and marketplace design---even though drivers could in principle game the system by expressing preferences for locations associated with more highly compensated rides.  There are two key elements to our approach that both seem likely to provide practical insight beyond the specific framework and mechanism considered here: First, we recognize that respecting drivers' location preferences creates value, which can at least partially substitute for cash compensation. Then, we incentivize truthful location preference revelation through a variation on a revealed preference approach. PARM uses drivers' deviations from proposed dispatches to learn about their preference types---a driver who chooses to drive to $i$ instead of her assigned location is inferred to prefer location $i$ and subsequently faces the compensation profile of other drivers with that preference.
For this approach to work, it is important that preferences do not change frequently over the course of the day. Otherwise, it would be much harder to enforce incentive compatibility by tracking endogenous responses to dispatch assignments.

\section*{Acknowledgments}

We appreciate the helpful comments of Peter Frazier, Jonathan Hall, Hamid Nazerzadeh, Michael Stewart, and the Lab for Economic Design. Rheingans-Yoo gratefully acknowledges the support of a Harvard Center of Mathematical Sciences and Applications  Economic Design Fellowship, a Harvard College Program for Research and Science and Engineering Fellowship, and the Harvard College Research Program. Rheingans-Yoo also appreciates the academic inspiration and personal support of Ross Rheingans-Yoo.
Kominers gratefully acknowledges the support of National Science Foundation grant SES-1459912, as well as the CMSA's Ng Fund and Mathematics in Economics Research Fund. 
Ma gratefully acknowledges the support of a Siebel scholarship.

\newpage


\bibliographystyle{named}
\bibliography{ijcai19}

\newpage

\appendix

\noindent{}\textbf{\huge{Appendix}}

\bigskip

\noindent{}We provide in Appendix~\ref{appx:proofs} the proofs that are omitted from the body of the paper. We present in Appendix~\ref{appx:no_penalty} an example of an incentive issue when no penalty is imposed.
\section{Proofs} \label{appx:proofs}

\subsection{Proof of Lemma~\ref{lem:c_forms_eq}} \label{appx:proof_lem_c_forms_eq}

\lemcformseq*

\begin{proof}
The first term of the objective of \eqref{equ:rev_opt_rewritten} is the same as the income portion of the objective of \eqref{equ:rev_opt}. The constraints  on this optimization are exactly the equilibrium conditions (C2), (C3) and (C5). Thus, it suffices to devise a compensation scheme where driver income everywhere is exactly equal to outside option, i.e.~$\pi_i\supt=w$ (this will necessarily satisfy (C1) and (C4)). This will be revenue-optimal because by (C4) drivers cannot be making less than $w$. Consider compensation scheme \eqref{equ:compensation}:
\begin{align}
    c_{ij} \supt \hspace{-0.2em} =
       \hspace{-0.2em} W \hspace{-0.2em} -I\cdot\one{i=\dtype}
        \notag
\end{align}
The second term means that any idiosyncratic utility a driver gets is extracted by the platform, so any dispatched driver makes exactly $W$ in that period. Thus, $\pi_i\supt=w=\frac{W}{1-\beta}$ exactly when probability of dispatch at every location is 1, i.e.~ 
\[
 x_{i} \supt =  \sum_{j \in \location} f_{ij} \supt + y_{ij} \supt, ~\forall i \in \location,~ \forall \dtype \in \location
\]
This is equivalent to the second-to-last constraint of \eqref{equ:rev_opt_rewritten}, so any optimal solution to \eqref{equ:rev_opt_rewritten} will have $\pi_i\supt=w,\ \forall i\in\location$, so (C1) and (C4) are satisfied. Thus an optimal solution to \eqref{equ:rev_opt_rewritten} corresponds to an optimal solution to \eqref{equ:rev_opt} under compensation scheme \eqref{equ:compensation}.
\end{proof}

\subsection{Proof of Theorem~\ref{thm:parm_IC}}
\label{appx:proof_thm_parm_IC}
\thmparmic*
\begin{proof}
By Theorem~\ref{thm:IC_assuming_no_deviation}, if a driver is going to always provide service, she cannot profitably misreport. By Lemma~\ref{lem:strat}, if a driver reports her type truthfully, it is not profitable for her to strategically decline to provide service. So for a misreport to be profitable, it must be paired with post-reporting deviation (strategically declining to provide service). Lemma~\ref{lem:strat} characterizes what this deviation must be: providing service everywhere except her 
reported preferred location. There she instead drives to her true preferred location. 

Thus to show incentive compatibility without penalty, it suffices to show this strategy is not more profitable than truthfully reporting and always providing service. Recall how we calculate penalties:
\begin{align}
        \pi_i^{k\ra\dtype}=&\one{i\neq\dtype}\left(W+\delta\sum_j \dfrac{f_{ij}\supt+y_{ij}\supt}{x_i}\pi_j^{k\ra\dtype}\right)\notag\\
        &+I\cdot\one{i=k}+\one{i=\dtype}(\delta w - P^{k\ra\dtype}),\forall i,k\notag;\\
        w=&\sum_i \dfrac{x_i\supt}{\sum_j x_j\supt}\pi_i^{k\ra\dtype},\forall\notag k\label{equ:misreport_total}.
\end{align}
Where $P_\dtype\triangleq \max \{ \max_{k \in } \{ P^{k \ra \dtype} \},~ 0\}$. If we had set penalties for switching types from $k$ to $\dtype$ to be $P^{k\ra\dtype}$, a driver of type $k$ pretending to be type $\dtype$ following the strategy from Lemma~\ref{lem:strat} would get utility as defined in \eqref{equ:misreport_total}. The first term is the compensation from providing service at locations that are not $\dtype$. The second term is idiosyncratic utility from being at her favorite location. The third term is the utility from declining to provide service at $\dtype$ and driving to $k$, after which she pays the penalty, the platform updates her type, and she makes $w$ afterwards. Then with random initialization, by \eqref{equ:misreport_total} the driver's expected utility is exactly $w$. Because we set $P_\dtype=\max_k P^{k\ra\dtype}$, The equality in \eqref{equ:misreport_total} is an inequality, but still the driver's expected utility is no more than $w$, which is what she would make by reporting truthfully and always providing service. Therefore, it is a best response for each driver to report truthfully and always provide service.
\end{proof}

\subsection{Proof of Lemma~\ref{lem:strat}}
\label{appx:proof_lemma_strat}
\lemstrat*
\begin{proof}

Suppose we have a driver of type $\dtype$ who has reported she was of type $j$. If $j=\dtype$, then the driver makes $W$ every period she provides service. If $j\neq\dtype$, the driver makes $W+I$ every period she provides service at $\dtype$ (extra utility but not paid less), makes $W-I$ every period they provide service at $j$ (paid less but not extra utility), and makes $W$ every period elsewhere.

First, we will show that the driver will never choose to not provide service and drive to $\dtype$. This move will not change how the platform treats the driver in the future. So in either case ($j=\dtype$ and $j\neq\dtype$) the driver has lost $W$ in income and then relocated to the location where she makes weakly least in the network. This is not profitable, so the driver will not choose to not provide service and drive to $\dtype$.

Next, we will show that if a driver chooses to not provide service, she will drive to her true preferred location. Suppose this driver of type $\dtype$ declines to provide service. As established in the previous paragraph, she will not relocate to her previously reported preferred location, so she may pay a penalty, but the penalty is the same no matter where she relocates to. she has the choice of where to drive in the network and will choose the location $i$ that maximizes her expected lifetime earnings \textit{given that she will be treated as type-$i$ in the future}. If location $i$ is such a best-response choice, then for every future deviation, the symmetry of the situation implies that $i$ will be a best-response choice then too. So we can assume that for all future deviations, the driver drives to $i$.

In the case that $i = \dtype$, the driver expects to make $W = w(1-\delta)$ in every period after arriving at location $\dtype$. Now suppose that $i \neq \dtype$. The driver gets utility $W + I$ every period she is at location $\dtype$, and $W-I$ every period she is at $i$, and $W$ everywhere else. With $\delta\rightarrow 1$, this means choice of $i \neq \dtype$ can only be better than $\dtype$ if she spends more time at $\dtype$ than $i$ before her next deviation, at which point she faces the same choice. However, we know that $x_{i}^{(i)} \geq x_{\dtype}^{(i)}$, which by the Ergodic Theorem implies that a driver treated as type-$i$ will spend on average at least as much time at location $i$ than location $\dtype$. This implies that if the deviant driver starts at location $i$, the expected number of times she visits location $\dtype$ before returning to $i$ is at most 1 (otherwise she would spend more time on average at $\dtype$ than $i$). So for any $t$, $\E{N_{\dtype}^{(i)}(t)} \leq \E{ N_{i}^{(i)}(t)}$, where $t$ is the number of time periods and $N_{j}^{(i)}(t)$ is the number of periods in which the driver is in location $j$. So there is no time in the future by which point the driver expects to have been at $\dtype$ more than $i$ if she follows platform instructions. The driver can deviate from platform instructions, but as established previously, without loss of best response she will drive to location $i$, which puts in the same position as before.

Finally, we will show that a driver will choose to provide service everywhere if $j=\dtype$ and will choose to provide service everywhere except possibly $j$ if $j\neq\dtype$. We have already established that a driver's best relocation is her true preferred location but that a driver will not decline to provide service and drive to her previously reported preferred location. It follows immediately that truthful drivers (those whose previously reported preferred location is their true preferred location) will always provide service. So we assume $j\neq\dtype$. Because a driver's post-deviation relocation is her preferred relocation, she will make $W$ in every period thereafter, after paying the platform a penalty. Before deviation, the driver makes $W+I$ at $\dtype$, $W-I$ at $j$, and $W$ everywhere else. So the only location she might not want to provide service at is $j$. Everywhere else her earnings are the same as post-deviation, and the only time it will be less is when at $j$. So the driver will always choose to provide service at locations other than $j$.
\end{proof}

\subsection{Proof of Theorem~\ref{thm:inf_supply}} \label{appx:proof_thm_inf_supply}

\thminfinitesupply*

\begin{proof}
We show full-information first-best revenue by showing the IC constraint does not bind. Consider an optimal solution to \eqref{equ:rev_opt_rewritten} without IC constraint \eqref{equ:ic_constraint}. The flow constraint allows us to decompose the flow of type-$\dtype$ drivers into cycles with various mass. Suppose one of these cycles does not go through location $\dtype$. Then it is optimal to replace that driver flow with drivers of type-$j$, for location $j$ in the cycle. (this can necessarily happen because there is no supply constraint). The new drivers do exactly what the old drivers did, so the demand met is exactly the same and the flow constraints are still satisfied. However, they end up in their preferred location strictly more, so the last term of the objective strictly increases. So the previous solution was not optimal. So all the cycles of type-$\dtype$ driver flow go through location $\dtype$. This means $x_{\dtype}\supt\geq x_j\supt,\forall \dtype,j$ because all flow of $\dtype$ drivers through $j$ also goes through $\dtype$. So any optimal solution to \eqref{equ:rev_opt_rewritten} naturally satisfies the IC constraint \eqref{equ:ic_constraint}. So imposing the IC constraint does not lead to an objective loss.

We show no penalty is necessary to achieve incentive compatibility by showing that the strategy described in Lemma~\ref{lem:strat} is not profitable. Consider $\delta\ra 1$ and suppose a driver of type $\dtype$ reports he is of type $i$. Every time he visits location $\dtype$, he will visit location $i$ before returning to $\dtype$. Otherwise, there would be some non-negative flow in a cycle through $\dtype$ but not through $i$, which the platform could more optimally fill with reported 
type-$\dtype$ drivers. So the driver will always eventually visit $i$ and will never visit location $\dtype$ more than once before doing so. The driver makes $W+I$ at location $\dtype$, then $W$ every period before giving a ride to $i$, then $0$ at location $i$, before relocating to $\dtype$ and being thereafter treated as $\dtype$, making $W$ in every subsequent period. So compared to truthful behavior, the driver makes a maximum of $I$ extra at $\dtype$ and loses a minimum of $W\geq$ at location $i$ before revealing his true type and making $W$ every period thereafter. With $\delta\ra 1$, this strategy is not more profitable than truthful reporting and always providing service. So even without a penalty, PARM is incentive compatible. 
\end{proof}

\subsection{Proof of Lemma~\ref{lem:self_fill}} \label{appx:proof_lem_sef_fill}

\lemmaselffill*

\begin{proof}
Consider an arbitrary $i,\dtype$. First notice that in any optimal solution, $y_{ii}\supt=y_{\dtype\dtype}\supt=0$ because these are the drivers are not fulfilling rides and staying in the same location, which adds cost to the system without fulfilling demand or helping to satisfy any of the constraints. Next, if $f_{ii}\supt>0$ and $f_{\dtype\dtype}^{(j)}>0$ for $j\neq \dtype$, then it is more optimal to switch drivers of type $j$ going $\dtype\ra \dtype$ with drivers of type $\dtype$ going $i\ra i$. This fulfills the same demand, but with less cost because you have drivers of type $\dtype$ at location $\dtype$ more. This violates the optimality of the original solution, so we can assume that $f_{ii}\supt>0\im f_{\dtype\dtype}^{(j)}=0\ \forall j\neq \dtype$ and that $\exists j\neq \dtype$ s.t. $f_{\dtype\dtype}^{(j)}>0\im f_{ii}\supt=0$. So we find ourselves in one of two cases: \\

\noindent{}{\em Case 1.} $f_{ii}\supt=0$: If this is the case, then $f_{ii}\supt\leq f_{\dtype}\supt$ because $f_{\dtype\dtype}\supt\geq 0$.

\noindent{}{\em Case 2.} 
$f_{ii}\supt>0$: Then $f_{\dtype\dtype}^{(j)}=0\ \forall j\neq \dtype$. Suppose $f_{ii}\supt>f_{\dtype\dtype}\supt$. Then $\sum_j f_{ii}^{(j)}>\sum_j f_{\dtype\dtype}^{(j)}\im p_{ii}<p_{\dtype\dtype}$ because the demand pattern is symmetric. I will show it is more optimal for the platform to raise $p_{ii},f_{\dtype\dtype}\supt$ and lower $p_{\dtype\dtype},f_{ii}\supt$ by infinitesimal amounts. Making substitutions and differentiating, we get the derivative of the first term of the objective with respect to $p_{ii}$ is given as follows:
    \[
    \dfrac{d}{dp_{ii}}\obj_1=1-2\dfrac{p_{ii}}{\theta_i\alpha_{ii}}
    \]
    We know that $\theta_i\alpha_{ii}=\theta_t\alpha_{tt}$ and that $p_{ii}<p_{\dtype\dtype}$. So
    \[
    \dfrac{d}{dp_{ii}}\obj_1>\dfrac{d}{dp_{\dtype\dtype}}\obj_1
    \]
    So the platform can make a marginal increase in the first term of the objective by raising $p_{ii}$ and lowering $p_{\dtype\dtype}$ by infinitesimal amounts, shifting an infinitesimal amount of $f_{ii}\supt$ to $f_{\dtype\dtype}\supt$ to make the market clear. The same number of drivers are in the system, so it does not change the second term of the objective. And it increases the third term in the objective because we just shifted drivers of type $\dtype$ to only be at location $\dtype$. So this shift brings us to a more optimal solution, which is a contradiction. So without the IC constraint imposed, $f_{ii}\supt\leq f_{\dtype\dtype}\supt\forall i,\dtype.$
\end{proof}

\subsection{Proof of Lemma~\ref{lem:bilateral}}
\label{appx:proof_lem_bilateral}

\lemmabilateral*

\begin{proof}
Consider an optimal solution to the platform's optimization problem. Now set $p_{ij}=p_{ji}$ and $f_{ij}\supt=\frac{1}{2}(f_{ij}\supt+f_{ji}\supt),y_{ij}\supt=\frac{1}{2}(y_{ij}\supt+y_{ji}\supt)$. The same number of drivers of each type is used, so the supply constraint is still satisfied. Flow into and out of a location before was the same, so it will be the same now too, so the flow constraint is satisfied. And the symmetry of the demand pattern implies that $\theta_i\alpha_{ij}=\theta_j\alpha_{ji}$ and so if the market cleared before, setting prices and flows to be the same will still clear the market.

So all we need to show is that revenue is weakly better when $p_{ij}=p_{ji}$, holding number of drivers used (and thus $p_{ij}+p_{ji}$) constant. Let $p_{ij}+p_{ji}=q$. Then isolating the part of the objective that changes under this switch yields:
\begin{align*}
    &\theta_i\alpha_{ij}p_{ij}(1-p_{ij})+\theta_j\alpha_{ji}p_{ji}(1-p_{ji})\\
    =&\theta_i\alpha_{ij}p_{ij}(1-p_{ij})+\theta_i\alpha_{ij}(q-p_{ij})(1-(q-p_{ij}))\\
    \frac{\partial}{\partial p_{ij}}=&\theta_i\alpha_{ij}\Big(1-2p_{ij}-1+2q-2p_{ij}\Big)\\
    0=&2q-4p_{ij}\\
    p_{ij}=&\frac{q}{2}
\end{align*}
So holding $p_{ij}+p_{ji}$ constant, equal prices is optimal. So any solution to the optimization problem can be made without loss of optimality into a solution with $f_{ij}\supt=f_{ji}\supt, y_{ij}\supt=y_{ji}\supt\forall i,j,\dtype$. Having drivers flowing back and forth without giving rides is sub-optimal, so this implies $y_{ij}\supt=y_{ji}\supt=0$. 
\end{proof}

\subsection{Proof of Theorem~\ref{thm:fb_symmetric_demand}}
\label{appx:proof_thm_symm_demand}

\thmfbsymmdemand*

\begin{proof}
We will show full-information first-best revenue by showing the IC constraint does not bind. By Lemma~\ref{lem:self_fill}, the symmetry of the demand pattern means that $f_{ii}\supt\leq f_{\dtype\dtype}\supt\ \forall i,\dtype$. By Lemma~\ref{lem:bilateral}, we can restrict our attention to solutions where all flow of drivers is bilateral and there are no floating drivers. Suppose in an optimal solution $x_i\supt>x_\dtype\supt$. This implies $\sum_{j\neq i}f_{ij}\supt>\sum_{j\neq \dtype}f_{\dtype j}\supt$, which in turn implies that there exists a location $j$ s.t.~$f_{ij}\supt=f_{ji}\supt>f_{\dtype j}\supt=f_{j\dtype}\supt$. We find ourselves in one of two cases:
\begin{enumerate}[1.]
    \item Exists $k\neq \dtype$ s.t.$f_{\dtype j}^{(k)}=f_{j\dtype}\supt>0$. Then we can switch some type-$\dtype$ drivers going $i\rightarrow \dtype\ra i$ with some type k drivers going $\dtype\ra j \ra \dtype$. Formally, we make $f_{\dtype j}^{(k)},f_{j\dtype}^{(k)},f_{ij}\supt,f_{ji}\supt$ smaller by $\epsilon$ and make $f_{\dtype j}\supt,f_{j\dtype}\supt,f_{ij}^{(k)},f_{ji}^{(k)}$ larger by $\epsilon$. All the same demand is filled, the same number of drivers of each type are used, and we've swapped drivers in a way that preserves the flow constraint, so all the constraints are satisfied. We've just strictly increased $x_\dtype\supt$ while weakly increasing $x_{kk}$ (increase in the case that $k=i$), which decreases our cost. This improvement is a contradiction with the optimality of the original solution.
    \item DNE $k\neq \dtype$ s.t. $f_{\dtype j}^{(k)}=f_{j\dtype}\supt>0$. Then $\sum_k f_{ij}^{(k)}=\sum_k f_{ji}^{(k)}>\sum_k f_{\dtype j}^{(k)}=\sum_k f_{j\dtype}^{(k)}$. The symmetry of the demand pattern then implies that $p_{ij}=p_{ji}<p_{\dtype j}=p_{j\dtype}$, and so it is an objective improvement to raise $p_{ij},p_{ji},f_{\dtype j}\supt,f_{j\dtype}\supt$ by $\epsilon$ and lower $p_{\dtype j},p_{j\dtype},f_{ij}\supt,f_{ji}\supt$ by $\epsilon$. This violates the optimality of the original solution.
\end{enumerate}
So there exists an optimal bilateral-flow solution, and this solution must have $x_i\supt\leq x_\dtype\supt\forall i,\dtype$.

We will show no penalty is necessary to achieve incentive compatibility by showing that the deviation described in Lemma~\ref{lem:strat} is not profitable for the solution constructed in Lemma~\ref{lem:bilateral}. Consider $\delta\ra 1$. We know that for all $i,j,\dtype$ that $f_{i\dtype}\supt\geq f_{ij}\supt$, so a driver of reported type 
$\dtype$ is more likely to be sent to $\dtype$ than $j$ no matter what location she is at. So, with random initialization of drivers, a driver of type $j$ pretending to be type $\dtype$ is more likely to be sent to $\dtype$ before than $j$ than the other way around. Even if sent to $j$ first, she is more likely than not to visit $\dtype$ before visiting $j$ again. When the driver visits $j$, she declines to provide service, giving up $W$ in income and then relocating to $\dtype$, making $W$ in every subsequent period. When the driver visits $\dtype$, she gets $I\leq W$ in extra idiosyncratic utility. However, the probability she visits $i$ (and lose $W)$ before $\dtype$ is at least $\frac{1}{2}$. So from the beginning of the game to the end of her first period in location $\dtype$ or $i$, she loses more in expectation than she gains. Even if she is sent to $\dtype$ first, the probability she visits $i$ (losing $W$ and then revealing her true preference) before returning to $\dtype$ and making another $I$ is less than $\frac{1}{2}$, so she is still making less than if she had truthfully reported and always provided service.
\end{proof}

\section{Incentive Problem Without Penalty}
Figure 10 illustrates driver flow for $\theta_1=\theta_0=100$, $\alpha=[(1.0,0.0),(0.0,1.0)]$, $s=(200,5)$, $I=0.2W$. Importantly, there is no demand between locations, so once a driver is assigned to a location, they will continue giving rides there and never visit another location, unless they decline to provide service in a period and relocate. Then a type 1 driver has a useful deviation. She reports that she is type 0, then if she is sent to location 1, she provides service in every period. Thus she prefers location 1 but is being paid as if she does not, and is never sent to location 0 (where she would be paid less) because there is no demand between locations. If this happens, she gets utility $w+\frac{I}{1-\delta}=48$ over her lifetime. If instead she is sent to location 0, she does not provide service the first period and relocates to location 1, after which she provides service and the platform treats her as type 1. If this happens, she gets utility $w-W=39.4$ over her lifetime. So her expected lifetime utility from this deviation is $\frac{34}{59}\cdot 39.4+\frac{25}{59}\cdot 48\approx 43.16$, which is greater than her utility from reporting truthfully and providing service, which is $w=40$. So this deviation is useful. The key thing going on here is that although more type 0 drivers are at location 0 than 1, an individual type 0 driver might spend their entire lifetime at location 1, which incentivizes type 1 drivers to misreport, taking the risk of one lost period of income in order to get a lifetime of idiosyncratic utility. In this case, the markov chain describing the movement of a type 0 driver is disconnected, but this issue also occurs when the markov chain is only very weakly connected and the driver is balancing one period of income versus many periods of idiosyncratic utility. The penalty is calculated such that if a driver does the deviation described here, the extra loss when they decline to provide service makes the deviation not useful. And in our proof of Theorem~\ref{thm:parm_IC}, we show that this deviation is the best of all non-truthful strategies, so the penalty ensures incentive compatibility.
\label{appx:no_penalty}
\begin{figure}[htpb!]
    \centering
    \begin{subfigure}[b]{ \flowFigSideWidth \textwidth}
        \centering
        \begin{tikzpicture}
            [ font = \flowFigFontSize,
              >/.tip={Triangle[length=7.5pt,width=\flowFigArrowSize,bend]},
              line width=1pt,
              my circle/.style={minimum width=\flowFigCrlSize, circle, draw},
              my label/.style={above=\flowFigLabelAbove, anchor=mid}
            ]
            \node (0) [my circle] {$0$};
            \node (1) [my circle, right=\flowFigCircleGap of 0] {$1$};
            \path [->, draw] (1.north east) arc (\flowFigArcParam) node [midway, my label] {$25$};
            \path [->, draw] (0.north east) arc (\flowFigArcParam) node [midway, my label] {$34$};
        \end{tikzpicture}
        \caption[Network2]%
        {{\small Type 0: PARM}}    
        \label{fig:mean and std of net14_first}
    \end{subfigure}
    ~
    \begin{subfigure}[b]{ \flowFigMidWidth \textwidth}   
        \centering 
        \begin{tikzpicture}
            [ font = \flowFigFontSize,
              >/.tip={Triangle[length=7.5pt,width=\flowFigArrowSize,bend]},
              line width=1pt,
              my circle/.style={minimum width=\flowFigCrlSize, circle, draw},
              my label/.style={above=\flowFigLabelAbove, anchor=mid}
            ]
            \node (0) [my circle] {$0$};
            \node (1) [my circle, right=\flowFigCircleGap of 0] {$1$};
            \path [->, draw] (1.north east) arc (\flowFigArcParam) node [midway, my label] {$5$};
            \path [->, draw] (0.north east) arc (\flowFigArcParam) node [midway, my label] {$0$};
        \end{tikzpicture}
        \caption[]%
        {{\small Type 1: PARM}}    
        \label{fig:mean and std of net34_second}
    \end{subfigure}
    ~
    \caption[Flow]
    {\small Rider trip and idle driver flows, with $\theta=(100,100)$, $\alpha=[(1.0,0.0),(0.0,1.0)]$, $s=(200,5)$, and $I=0.2W$.}
    \label{fig_no_penalty}
\end{figure}
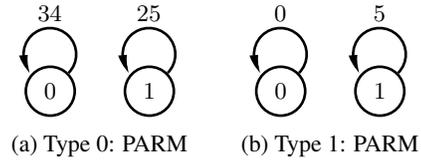

\end{document}